\documentclass{article}

\usepackage{amsfonts}
\usepackage{amsmath}
\usepackage{amssymb}
\usepackage{graphicx}
\usepackage{verbatim}

\usepackage{amsfonts}
\usepackage{amsmath}
\usepackage{amssymb}
\usepackage{enumerate}

\newcommand{\myset}[2]{ \left\{ #1 \left| \, #2 \right. \right\} }

\newcommand{\diam}{\mathrm{diam}}

\newcommand{\ignore}[1]{}

\newcommand{\N}{\mathbb{N}}
\newcommand{\Z}{\mathbb{Z}}

\newcommand{\R}{\mathbb{R}}

\newcommand{\constr}{{\mathrm{constr}}}

\newcommand{\dimh}{\mathrm{dim}_\mathrm{H}}

\newcommand{\cdim}{\mathrm{cdim}}

\newcommand{\regSS}{S^\infty}

\newcommand{\C}{\mathbf{C}}

\newcommand{\K}{{\mathrm{K}}}

\usepackage{ifthen}

\newtheorem{theorem}{Theorem}[section]

\newtheorem{property}[theorem]{Property}

\newtheorem{corollary}[theorem]{Corollary}
\newtheorem{claim}[theorem]{Claim}

\newenvironment{definition}
{ {\noindent {\bf Definition.}} } {  }
\newtheorem{example}[theorem]{Example}
\newenvironment{example*}[1]
{ {\noindent {\bf Example #1.}} } {  }
\newenvironment{claim*}[1]
{ {\noindent {\bf Claim #1.}} } {  }
\newenvironment{theorem*}
{ {\noindent {\bf Theorem.}} } {  }

\newenvironment{proof}[1][xyzzy]
{
{\noindent {\bf Proof}%
\ifthenelse{\equal{#1}{xyzzy}}{{\bf .}}{~(#1).}} } { \hfill $\Box$ }

\newenvironment{proofof}[1]
{
{\noindent {\bf Proof #1.}%
}} { \hfill $\Box$ }

\newcommand{\cB}{\mathcal{B}}
\newcommand{\cBn}{\mathcal{B}_n}
\newcommand{\upto}{\upharpoonright}

\title{Effective Hausdorff dimension in general metric spaces}
\author{Elvira Mayordomo\footnote{Dept. de Inform\'atica e Ingenier\'{\i}a
de Sistemas, Instituto de Investigaci\'on en Ingenier\'{\i}a de Arag\'on (I3A), Universidad de Zaragoza, Zaragoza, SPAIN.
elvira@unizar.es.  Part of this work was done during a visit to Iowa State University, supported by NSF Grants 1247051 and
1545028.}}

\date{\today}

\begin{document}

\maketitle

\begin{abstract}
We introduce the concept of effective dimension for a wide class of metric spaces that are not required to have a
computable measure. Effective dimension was defined by Lutz in (Lutz 2003) for Cantor space and has also been extended to
Euclidean space. Lutz effectivization uses the concept of gale and supergale, our extension of Hausdorff dimension to other
metric spaces is also based on a supergale characterization of dimension, which in practice avoids an extra quantifier
present in the classical definition of dimension that is based on Hausdorff measure and therefore allows effectivization
for small time-bounds.

 We present here
the concept of constructive dimension and its characterization in terms of Kolmogorov complexity, for which we extend the
concept of Kolmogorov complexity to any metric space defining the Kolmogorov complexity of a point at a certain precision.
Further research directions are indicated.

\end{abstract}

\section{Introduction}

 Effective dimension in Cantor space
was defined by Lutz in \cite{DCC, DISS}\ in order to quantitatively
study complexity classes \cite{FGCC}. The connections of effective
dimension with Information Theory \cite{EFDAIT}, in particular with
Kolmogorov complexity and compression algorithms, some of them
suspected even before the definition of effective dimension itself
(\cite{Ryab84, Ryab86, Stai93, Stai98, CaiHar94} and more recently
for symbolic dynamical systems in \cite{Simp11}), have lead to very
fruitful areas of research including those within Algorithmic
Information theory \cite{DowHir10}.

In this paper we will explore the definition of effective dimension
for general metric spaces. The long term purpose of this line of
research is to find more and easier dimension bound proofs in those
spaces, while the connections with Information Theory already
suggest further developments.

The original definition of effective dimension was done in Cantor
space which is the set of infinite binary sequences with the usual
longest-common-prefix-based metric. The spaces of infinite sequences
over other finite alphabets have also been  explored, for instance
the case of Finite-State effectivity is particularly interesting
with this variation \cite{FSD}. 
Also the Euclidean space $\R^n$ has been considered in this context
by several papers that go back to fractal geometry, starting in
\cite{DPSSF}\ and including recent progress on the Kakeya conjecture
\cite{LutLut16}, and the space $h^{\omega}$, a generalization of
Cantor space, is studied in \cite{GreMil11}. Our approach
generalizes all the cited existing ones and does not require the
underlying space to have a computable measure.\footnote{Several
authors \cite{Gacs05,HoyRoj09,Miya14}\ have defined Martin-L\"of
algorithmical randomness in certain computable metric spaces, with
the obvious requirement of a computable Borel measure that we avoid
for effective dimension.}


Gales and supergales, introduced by Lutz in \cite{DCC}, are
intuitively betting strategies in a guessing game on the elements of
Cantor space. They allow the interpretation of Hausdorff dimension
in terms of prediction and provide natural effectivizations of
dimension by restricting the computability and resource-bounds used
in the computation of these betting strategies. In terms of the
complexity of the definition, the effectivization based on a gale
characterization of dimension avoids an extra quantifier on the
covers that would be present in effectivizations based on Hausdorff
measure and makes it possible to attain low time-bounds
effectivizations such as polynomial-time dimension.

We introduce here the concept of a nice cover of a metric space. A nice cover can simulate very closely any of the covers
required in the definition of Hausdorff dimension, while it allows simple representations of the points in the space and
the use of gales as betting games on those representations. Spaces with nice covers can
be fairly general and are locally separable. 


We then characterize Hausdorff dimension using supergales for any
metric space with a nice cover. This characterization allows the
definition of effective dimension by restricting the family of
supergales that can be used.


We  present here an initial step in this effectivization direction
by introducing the definition of constructive dimension on a metric
space. We then extend the concept of  Kolmogorov complexity to any
metric space and define the Kolmogorov complexity of a point at a
certain precision. We characterize constructive dimension in terms
of Kolmogorov complexity and prove the property of absolute
stability (that is, the fact that constructive dimension can be
pointwise defined). We present a few interesting examples of spaces
with a constructive dimension including all previously studied
cases.
 We finish
with a list of topics for further development.


\section{Preliminaries}

In this section we include the definition of Hausdorff diemnsion for
general metric spaces and include basic notation for strings and
sequences.

Let $(X, \rho)$ be a metric space. (From now on we will omit $\rho$
when referring to space $(X, \rho)$).

\begin{definition} The {\sl diameter\/} of a set $A\subseteq X$ is
\[\diam(A)=\sup \myset{\rho(x,y)}{x, y\in
A}.\] Notice that the diameter of a set can be infinite.
\end{definition}

\begin{definition} Let $A\subseteq X$. A {\sl cover of $A$\/} is $\mathcal{C}\subseteq
\mathcal{P}(X)$ such that $A\subseteq \cup_{U\in\mathcal{C}}
U$.\end{definition}

\begin{definition} Let $A\subseteq X$. $A$ is {\sl separable\/} if
there exists a countable set $S\subseteq A$ that is dense in $A$,
that is, for every $x\in A, \delta>0$ there is an $s\in S$ such that
$\rho(x,s)<\delta$.\end{definition}

\begin{definition} The {\sl ball\/} of radius $r>0$ about $x\in X$ is the set
$B(x,r)=\myset{y\in X}{ \rho(y,x) < r}$. \end{definition}

\begin{definition} A {\sl neighborhood\/} of $x\in X$ is a set $U$
that contains a ball about $x$.
\end{definition}

\begin{definition} A {\sl locally separable space\/} is $X$ such that every $x\in X$ has a
neighborhood that is separable.\end{definition}




\begin{definition} An {\sl isolated point in $X$\/} is $x\in X$ such
that there is a $\delta>0$ with $B(x,\delta)\cap X=\{x\}$.
\end{definition}

We will be interested in metric spaces that have no isolated points.
Notice that metric spaces consisting only of isolated points have
little interest for Hausdorff dimension (see definition below),
while Hausdorff dimension in general spaces can be analyzed by
restricting to non isolated points in the space.

We include the basic definitions of Hausdorff dimension. We refer
the reader to \cite{Falc03}\ for a complete introduction and
motivation.

 For each
$A\subseteq X$ and $\delta
>0$, we define the set of countable $\delta$-covers of $A$
\[
\mathcal{H}_{\delta}(A)= \myset{ \mathcal{U}}{ {\mathcal{U}\mbox{ is
a countable cover of } A \mbox{ and }\diam(U)<\delta \mbox{ for
every }U\in\mathcal{U}} }.
\]
We can now define ${H}_{\delta}^s(A)$ and ${H}^s(X)$
\begin{eqnarray*}
{H}_{\delta}^s(A)&=& \inf_{\mathcal{U}\in\mathcal{H}_{\delta}(A)}
\sum_{U\in\mathcal{U}}\diam(U)^s.\\
{H}^s(A)&=& \lim_{\delta\to 0}{H}_{\delta}^s(A).
\end{eqnarray*}
Notice that ${H}_{\delta}^s(X)$ is monotone as $\delta\to 0$ so
${H}^s(X)$ is well defined. It is routine to verify that  ${H}^s$ is
an outer measure \cite{Falc03}, ${H}^s$ is called the {\sl
$s$-Hausdorff measure}.

\begin{definition}
 (Hausdorff \cite{Haus19}). The
{\sl Hausdorff dimension\/} of $A\subseteq X$  is
\[\dimh(A)=\inf\myset{s\in[0,
\infty) }{ H^s(A)=0}.\]
\end{definition}

Notice that $\dimh(A)$ is nontrivial only when $A$ is bounded, that
is, $A$ has a finite diameter. In general there is an abuse of
notation for unbounded $A$ using $\dimh(A)=\sup_{B\subseteq A,
B\mbox{ bounded}}\dimh(B)$.

Let $\Sigma$ be a finite set. We denote as $\Sigma^*$ the set of
finite strings over $\Sigma$ and $\Sigma^{\infty}$ the set of
infinite strings over $\Sigma$. If $w\in\Sigma^*$ and
$x\in\Sigma^*\cup\Sigma^{\infty}$ then $w\sqsubseteq x$ denotes $w$
being a finite prefix of $x$. For $0\le i \le j$, we write $x[i
\ldots j]$ for the string consisting of the $i$-th through the
$j$-th symbols of $x$. We use $\lambda$ for the empty string. For
$n\ge 0$, we write $x\upto n$ for $x[0..n-1]$. ($x\upto 0$ denotes
again the empty string).

We use $< , >$ for a  {\sl pairing function\/} $< , >:\Sigma^*
\times \Sigma^*\to \Sigma^*$ that is injective, efficiently
computable and invertible. We also denote with $<w_1, \ldots, w_k>$
the $k$-fold composition of $< , >$ for each $k\in\N$.

\section{A supergale characterization of dimension in metric spaces}

\subsection{Nice covers}\label{secnc}

We introduce the concept of a nice cover for a metric space. A nice
cover allows well behaving representations of all points in the
space, and it will be the key to the supergale characterization of
Hausdorff dimension in the next subsection. 
Intuitively, a nice cover of $A$ is a sequence of covers of $A$ that
can closely simulate any Hausdorff cover of $A$.

We prove here that spaces with a nice cover are locally separable
and that  spaces with countable nice covers are separable. We also
include examples of spaces with nice cover that include all spaces
for which effective dimension has been defined so far.

Our goal is to generalize Lutz characterization  for Cantor space
\cite{DCC, DISS}. Notice that Cantor space has very good properties
not generally present in other spaces, namely a basis of { clopen}
intervals defined by finite prefixes (for $w\in\Sigma^*$,
$C_w=\myset{x}{w\sqsubseteq x}$) and a { Borel measure} from which
the distance is defined ($\rho(x,y)= \inf\myset{\mu(w)}{x,y\in
C_w}$). These two facts make prediction and compression arguments
natural and simpler, our aim is to have a prediction (and later on
compression) setting for more general metric spaces.

Let $X$ be a metric space without isolated points.

\begin{definition} 
A {\sl nice cover of $X$} is a sequence $(\mathcal{B}_n)_{n\in\N}$
with $\mathcal{B}_n\subseteq \mathcal{P}(X)$  for every $n$ and such
that the following hold
\begin{description}

\item[(A1)] (Decreasing monotonicity) For every $n\in\N$, $U\in \mathcal{B}_n$, $|\{V\in \mathcal{B}_{n+1}, V\subseteq
U\}|<\infty$. 

\item[(A2)]\label{unique} (Increasing monotonicity) For every $n\in\N$, $U\in
\mathcal{B}_n$, $m<n$, there is a unique  $V\in \mathcal{B}_{m}$
such that $U\subseteq V$.
\item[(A3)]\label{nonnull} (Layerwise nonnull) For every $n\in\N$,
$\inf_{U\in\mathcal{B}_n}\diam(U) >0$.
\item[(A4)]\label{ccover} ($c$-cover) There is a $c\in\N$ such that  for every $A\subseteq {X}$ with
$0<\diam(A)<1$  there exists $\{U_1, \ldots, U_c\}
\subseteq\cup_{n}\mathcal{B}_{n}$ a cover of $A$, with
$\diam(U_i)<c\cdot \diam(A)$ for every $i$.
\end{description}
\end{definition}

Notice that the above definition does not even require the elements
of each cover $\cBn$ to be open or disjoint.


\begin{theorem} If $X$ has a nice cover then $X$ is
locally  separable.
\end{theorem}

\begin{theorem} If $X$ has a countable nice cover then $X$ is separable. \end{theorem}


We next include several examples of spaces with a nice cover
including all spaces for which Hausdorff dimension has been
effectivized. Our first example is the generalization of Cantor
space to any finite alphabet and to any Borel probability measure.

\begin{example}

Consider the set $\Sigma^{\infty}$ of all infinite sequences over a
finite alphabet $\Sigma$ with the metric based on a positive and
nonatomic Borel probability measure.

Let $\nu: \Sigma^*\to [0,1]$ be a Borel probability measure  on
$\Sigma^{\infty}$ (that is, $\nu(\lambda)=1$ and for each
$w\in\Sigma^*$ $\nu(w)=\sum_{a\in\Sigma}\nu(wa)$) such that $\nu$ is
positive (that is, $\nu(w)>0$ for every $w\in\Sigma^*$) and
nonatomic (that is, $\lim_n\nu(x[0..n-1])=0$ for every
$x\in\Sigma^{\infty}$). Then $\nu(w)$ is the measure of cylinder
$\C_w=\myset{x\in\C}{w\sqsubseteq x}$), and the distance is defined
as $\rho(x,y)=\inf\myset{\nu(w)}{x,y \in \C_w}$.

This space has a nice cover formed by cylinders $\C_w$. We take
$\mathcal{B}_{n}=$ $\myset{\C_w}{w\in\{0,1\}^n}$ for each $n\in\N$.
(In this case condition (A4) is met for $c=1$ by considering the
longest $w$ such that $A\subseteq \C_w$, we can prove that
$\diam(A)=\diam(\C_w)=\nu(w)$).

Notice that we place no restrictions on the probability measure
$\nu$, so this is a strict generalization of the case considered in
\cite{DPSSF}, where $\nu$ needed to be ``strongly positive'' as
defined in \cite{DPSSF}. This is because of the good properties of
Cantor space.

 Particular cases of this metric space are widely
used in connection with coding and compression algorithms, and also
to represent symbolic dynamical systems.

\end{example}

We will now move to Euclidean space.

\begin{example}

For $m\in\N$, consider the set $\R^m $ with the metric based on a strongly positive Borel measure $\nu\in\Delta(\R^m)$.

We use the notation $I(\mathbf{z}, \mathbf{a}, n)$ ($\mathbf{z}\in\Z^m$, $\mathbf{a}\in\N^m$, $n\in\N$) for the $n$-diadic
interval defined by $\mathbf{z}, \mathbf{a}$, that is, $I(\mathbf{z}, \mathbf{a},
n)=[z_1+a_12^{-n},z_1+(a_1+1)2^{-n}]\times \ldots \times [z_m+a_m2^{-n}, z_m+(a_m+1)2^{-n}]$ (for every $i$, $a_i< 2^n$ to
avoid multiple representations).

$\nu$ is strongly positive if there is a $\delta>0$ such that $\nu(I(\mathbf{t}, \mathbf{b}, n+1))>\delta\cdot
\nu(I(\mathbf{z}, \mathbf{a}, n))$ for every $I(\mathbf{t}, \mathbf{b}, n+1)\subseteq I(\mathbf{z}, \mathbf{a}, n)$.

This space has a nice cover formed by intervals with dyadic rational extremes. We take
$\mathcal{B}_{n}=\myset{I(\mathbf{z}, \mathbf{a}, n)}{\mathbf{z}\in\Z^m, \mathbf{a}\in\N^m}$ for each $n\in\N$. (In this
case condition (A4) is met for $c= \max(1/\delta, 2^m)$). Notice that an interval $[a,b]$ can be covered by two dyadic
intervals $A, B$ of size $\nu(A)\le \nu([a,b])/\delta$, $\nu(B)\le \nu([a,b])/\delta$.)

\end{example}

\begin{example}

Let $h:\N\to \N-\{0,1\}$, we define $h^{\omega}=\Pi_{n\in\N}\{0, 1,
\ldots, h(n)-1\}$  with the metric based on a positive and nonatomic
Borel probability measure $\nu$.

For $n\in\N$, let $h^{n}=\Pi_{m<n}\{0, 1, \ldots, h(m)-1\}$. Let
$h^*=\cup_{n\in\N}h^{n}$.

Given a Borel probability measure $\nu: h^*\to [0,1]$ (where
$\nu(w)$ is the measure of cylinder $\C_w=\myset{x\in
h^{\omega}}{w\sqsubseteq x}$), the distance is defined as
$\rho(x,y)=\inf\myset{\nu(w)}{x,y \in \C_w}$.

This space has a nice cover formed by cylinders $\C_w$. We take
$\mathcal{B}_{n}=\myset{\C_w}{w\in h^n}$ for each $n\in\N$. (In this
case condition (A4) is met for $c=1$ by considering the longest $w$
such that $A\subseteq \C_w$, we can prove that
$\diam(A)=\diam(\C_w)=\nu(w)$).

This is a strict generalization of the case considered in \cite{GreMil11}.

\end{example}

Our last example has not been considered before in the context of effective dimension, and can have further interest in the
context of spaces of functions.

\begin{example}\label{polyn}
Let $n\in\N$. Let $P_n(0,1)$ be the set of polynomials  with real coefficients and degree less than or equal to $n$,
together with the metric $d(f, g) = \|f-g\|_{\infty}=\sup_{x \in (0,1)} |f(x) - g(x)|$, for $f, g \in P_n(0,1)$.

This space has a nice cover defined from dyadic coefficient polynomials as follows. Let $\mathbf{z}\in\Z^{n+1}$,
$\mathbf{a}\in\N^{n+1}$, $k\in\N$, with $a_i< 2^k$ for every $i$,   \[DP(\mathbf{z}, \mathbf{a}, k)= \{ b_0+b_1x+\ldots +
b_nx^n\, |\, b_i\in\R, 0\le b_i-(z_i+a_i2^{-k})\le  2^{-k})\}.\] We take $\mathcal{B}_{k}=\myset{DP(\mathbf{z}, \mathbf{a},
k)}{\mathbf{z}\in\Z^{n+1}, \mathbf{a}\in\N^{n+1}}$ for each $k\in\N$. (In this case condition (A4) is met for $c=\max\{
2^n, 2(n+1)\}$).

\end{example}


\subsection{Supergale characterization of Hausdorff dimension}

In this subsection we prove a supergale characterization of
Hausdorff dimension for $X$ with a nice cover. 
Notice that each nice cover gives an equivalent characterization of
dimension. The definition of a supergale in terms of diameters gives
a new and interesting prediction setting for general metric spaces.

The concept of gale we introduce here is the natural extension of
the gales introduced in \cite{DCC}\ to spaces with nice covers,
while the flexibility on the metric spaces makes the proof of this
characterization quite more involved than the case of Cantor spaces
proven in \cite{DCC}. For instance we cannot assume anything about
the diameters of the covers used.

Let $X$ be a metric space with a nice cover, fix a nice cover
$(\mathcal{B}_n)_{n\in\N}$. Let $\cB=\cup_n \cBn$. For $n\in\N$, let
$\cB_{\ge n}=\cup_{m\ge n} \cB_m$.


\begin{definition} Given  $x\in X$, {\sl a $\cB$-representation of
$x$\/}
is a sequence $(w_n)_{n\in\N}$ such
that $w_n\in \mathcal{B}_n$ 
and $x\in \cap_n
w_n$.
\end{definition}

We denote with $\mathcal{R}(x)$ the set of $\cB$-representations of
$x\in X$.

A supergale is intuitively a strategy in a betting game on a
representation $(w_n)_{n\in\N}$ of an unknown $x\in X$.

\begin{definition} Let $s\in [0, \infty)$. An {\sl $s$-supergale $d$} is a
function $d: \mathcal{B}\to [0, \infty)$ such that the following
hold

\begin{itemize}
\item $\sum_{U\in \mathcal{B}_{0}}
d(U)\,\diam(U)^s < \infty$,
\item for every $n\in\N$, for every $U\in \mathcal{B}_n$ the following
inequality holds
\begin{equation}\label{eqg}
d(U)\,\diam(U)^s\ge \sum_{V\in \mathcal{B}_{n+1}, V\subseteq U}
d(V)\,\diam(V)^s. \end{equation}
\end{itemize}
An {\sl $s$-gale\/} is an $s$-supergale for which equation
(\ref{eqg}) holds with equality.
\end{definition}

\begin{definition} An $s$-supergale $d$ {\sl succeeds on $x\in X$}
if there is a $(w_n)_{n\in\N}\in \mathcal{R}(x)$,  such that
\[\limsup_n d(w_n)=\infty.\]
\end{definition}

\begin{definition} Let $d$ be an $s$-supergale. The success set of $d$ is
 \[\regSS[d]=\myset{x\in X}{d \mbox{ suceeds on }x}.\]\end{definition}


\begin{definition}
$\hat{\mathcal{G}}(A)=\myset{s}{\mbox{there is an }s\mbox{-supergale
}d \mbox{ with }A\subseteq\regSS[d]}$.
\end{definition}

We start with two useful properties of supergales. The first one
allows success to be defined independently of layers $\cB_n$ for
$n\in\N$, the second one is a generalization of Kraft inequality
that is used in Cantor space, that gives us bounds on  supergale
values.

\begin{property}\label{players} Let $d$ be a an $s$-supergale, if  there is a sequence in $\cB$\ $(y_m)_{m\in\N}$  such that
$\limsup_md(y_m)=\infty$ then $\cap_m y_m\subseteq \regSS[d]$.

\end{property}

\begin{proof} From $\limsup_md(y_m)=\infty$ and condition (\ref{eqg}) in the
definition of supergale we have that $\liminf_m \diam(y_m)=0$. In
particular there is a subsequence $(y_{m_k})_{k\in\N}$ with
$\lim_kd(y_{m_k})=\infty$ and $\lim_k \diam(y_{m_k})=0$. Therefore
using (A3) we have that for every $n$ there is a $k$ such that
$y_{m_k}\in \cB_{\ge n}$ which gives us $(w_n)_{n\in\N}$ such that
$w_n\in \mathcal{B}_n$ with $\cap_m y_m\subseteq \cap_n w_n$ and
$\limsup_nd(w_n)=\infty$.

\end{proof}

\begin{property}\label{cl4}(Generalization of Kraft inequality) Let $d$ be an $s$-supergale. Then for every $\mathcal{E}\subseteq\cB$
such that all sets in $\mathcal{E}$ are incomparable  we have that
\[\sum_{W\in\cB_0}d(W)\,\diam(W)^s\ge \sum_{V\in \mathcal{E}}
d(V)\,\diam(V)^s.\]\end{property}

\begin{proof}
Notice that $ \sum_{U\in \mathcal{B}_{0}} d(U)\,\diam(U)^s < \infty$
implies that $d(U)\,\diam(U)$ is nonzero on a countable subset of
$\cB_0$, and therefore $d(U)\,\diam(U)$ is nonzero on a countable
subset of $\cB$ because of (A1). Therefore let us assume that
$\mathcal{E}\subseteq\cB$ is countable.

We prove that for any finite subset $\mathcal{A}\subseteq \mathcal{E}$, \[\sum_{W\in\cB_0}d(W)\,\diam(W)^s\ge \sum_{V\in
\mathcal{A}} d(V)\,\diam(V)^s.\] This can be proven by induction on $n$ such that $\mathcal{A}\subseteq\cB_{\ge n}$ using
condition (\ref{eqg}) in the definition of supergale.
\end{proof}

Our next Theorem is the supergale characterization of dimension.
Since the metric used is not necessarily related to a Borel measure,
and  we can't assume a suitable basis of clopen (not even open) sets
the arguments in the proof are different from the Cantor space case
\cite{DCC}. The fact that we need to use diameters instead of a
probability measure makes it necessary to carefully examine each sum
of diameters for its good definition and finiteness.

\begin{theorem}\label{thsg}(Supergale characterization) Let $X$ be a metric space that has a nice cover, let $A\subseteq X$. Then
\[\dimh(A)= 
\inf\hat{\mathcal{G}}(A).\]\end{theorem}


\begin{proof}

Let $s>\dimh(A)$. Then for any $k\in\N$ there is a countable cover
of $A$, $\mathcal{C}_k$, such that $\sum_{U\in
\mathcal{C}_k}\diam(U)^s < 2^{-k}$ and $\diam(U)>0$ for each $U\in
\mathcal{C}_k$. (If necessary substitute each $U_n\in \mathcal{C}_k$
with $\diam(U_n)=0$ by a ball of radius $2^{-k/s-n/s-1})$.


Using property (A4) of nice covers we can get a cover
$\mathcal{E}_k\subseteq \mathcal{B}$ of $A$ such that
\[\sum_{W\in \mathcal{E}_k}\diam(W)^s < c^{1+s}\cdot 2^{-k}.
\]

 Let
$\mathcal{D}_k=\myset{U}{U\in\mathcal{E}_k\mbox{ and no proper
superset of }U\mbox{ is in }\mathcal{E}_k}$. Then $\mathcal{D}_k$ is
a countable cover of $A$ and \[\sum_{W\in \mathcal{D}_k}\diam(W)^s <
c^{1+s}\cdot 2^{-k}.
\]


 Define
$d_k: \mathcal{B}\to [0, \infty)$ as follows,

For $U\in\cB$, if $\diam(U)=0$ then $d(U)=1$.

If $U\in\cBn$ for $n>0$, $\diam(U)>0$, and there is $V\in\cB_{n-1}$
and $W\in \mathcal{D}_k$ with $U\subseteq V \subseteq W$, $U\ne V$
and $\sum_{\scriptstyle{U'\subseteq V, U'\in\cB_{n}}}\diam(U')>0$
then

\[d_k(U)=\frac{d_k(V)\,\diam(V)^s}{\sum_{\scriptstyle{U'\subseteq V,
U'\in\cBn}}\diam(U')^s}.\]

Otherwise, if  $U\in\cBn$, $\diam(U)>0$,

\[d_k(U)=\sum_{\scriptstyle W\in\mathcal{D}_k\cap \cB_{\ge n}, W\subseteq U}\frac{\diam(W)^s}{\diam(U)^s}.\]

\begin{claim}\label{cl1}$d_k$ is an $s$-supergale.
\end{claim}

\begin{proofof}{of Claim \ref{cl1}} Let $V\in\cB_{n-1}$
and \newline
 $\sum_{\scriptstyle{U'\subseteq V, U'\in\cB_{n}}}\diam(U')>0$.

If there is $W\in\mathcal{D}_k$ such that $V\subseteq W$ then

\begin{eqnarray*}\sum_{\scriptstyle{U\subseteq V,
U\in\cB_{n}}} d_k(U)\, \diam(U)^s &=& \sum_{\scriptstyle{U\subseteq
V, U\in\cB_{n}}} \frac{d_k(V)\,
\diam(V)^s}{\sum_{\scriptstyle{U'\subseteq V,
U'\in\cB_{n}}}\diam(U')^s}\diam(U)^s \\&=& d_k(V)\,
\diam(V)^s.\end{eqnarray*}

If for any $W\in\mathcal{D}_k$, $V\not\subseteq W$ then

\[d_k(V)= \sum_{\scriptstyle W\in \mathcal{D}_k\cap \cB_{\ge n-1}, W\subseteq V}
\frac{\diam(W)^s}{\diam(V)^s}.\]

Therefore,
\begin{eqnarray*}\sum_{\scriptstyle{U\subseteq V,
U\in\cB_{n}}} d_k(U)\, \diam(U)^s &=& \sum_{\scriptstyle{U\subseteq
V, U\in\cB_{n}}} \sum_{\scriptstyle W\in \mathcal{D}_k\cap \cB_{\ge
n}, W\subseteq U} \frac{\diam(W)^s}{\diam(U)^s} \diam(U)^s \\ &=&
\sum_{\scriptstyle{U\subseteq V, U\in\cB_{n}}} \sum_{\scriptstyle
W\in \mathcal{D}_k\cap \cB_{\ge n}, W\subseteq U} {\diam(W)^s}
\\ &\le&
 \sum_{\scriptstyle
W\in \mathcal{D}_k\cap \cB_{\ge n-1}, W\subseteq V} {\diam(W)^s}  =
d_k(V)\, \diam(V)^s,\end{eqnarray*} where the last inequality
follows from property (A2) of nice covers.

For every $U\in \cB_0$, we use the second part in the definition of
$d_k$. Therefore, using property (A2) of nice covers,
 \[\sum_{U\in \mathcal{B}_{0}}
d_k(U)\,\diam(U)^s \le  \sum_{\scriptstyle
W\in\mathcal{D}_k}{\diam(W)^s}< c^{1+s}\cdot 2^{-k}<\infty.\]

\end{proofof}

\begin{claim}\label{cl2}If $W\in \mathcal{D}_k$, $d_k(W)= 1$.
\end{claim}
\begin{proofof}{of Claim \ref{cl2}} If $W\in\cBn$, since all sets in
$\mathcal{D}_k$ are incomparable, we use the second part in the
definition of $d_k$ and
\[d_k(W)=\sum_{\scriptstyle W'\in\mathcal{D}_k\cap \cB_{\ge n}, W'\subseteq W}\frac{\diam(W')^s}{\diam(W)^s}= 1.\]
\end{proofof}

\begin{claim}\label{cl3}For every $k\in\N$, $U\in \cB$, with $\diam(U)>0$, $d_k(U)\le \frac{c^{1+s} \cdot
2^{-k}}{\diam(U)^s}$.
\end{claim}
\begin{proofof}{of Claim \ref{cl3}}

We prove by induction on $n-m$ that for every $n,m\in\N$ with $m<n$,
$U \subseteq V$ with $\diam(U)>0$, $U\in\cBn$ and $V\in\cB_m$,
\[d_k(U)\le \frac{d_k(V)
\,\diam(V)^s}{\diam(U)^s}.\]

By the definition of supergale, if $U\in \cBn$, $d_k(U)\le
\frac{d_k(U') \,\diam(U')^s}{\diam(U)^s}$ for $U'\in \cB_{n-1}$ with
$U\subseteq U'$. By induction $d_k(U')\le \frac{d_k(V)
\,\diam(V)^s}{\diam(U')^s}$ and therefore $d_k(U)\le \frac{d_k(V)
\,\diam(V)^s}{\diam(U)^s}$.

 For every $W\in \cB_0$ with $\diam(W)>0$, we use the
second part in the definition of $d_k$ and so $d_k(W)\le
\frac{c^{1+s} \cdot 2^{-k}}{\diam(W)^s}$.

Since for every $U\in \cB$ there is a  $W\in \cB_0$ with $U\subseteq
W$ we have that \[d_k(U)\le \frac{d_k(W)
\,\diam(W)^s}{\diam(U)^s}\le \frac{c^{1+s} \cdot
2^{-k}}{\diam(U)^s}.\]

\end{proofof}

We define next an $s$-supergale $d(U)= \sum_k  2^{k}d_{2k}(U)$.

By Claim \ref{cl3} $d$ is well-defined.

By Claim \ref{cl2}, if $W\in \mathcal{D}_k$, $d(W)\ge 2^k$. Since
for every $k$, $\mathcal{D}_{k}\subseteq \cB$ is a cover of $A$, by
Property \ref{players}\ we have that $A\subseteq \regSS[d]$ and
$s\in \mathcal{G}(A)$.

For the other direction, let $s\in \hat{\mathcal{G}}(A)$. Then there
exists an $s$-supergale $d$ such that $A\subseteq \regSS[d]$.

For each $k\in\N$ let \[\mathcal{C}_k= \myset{U}{\diam(U)>0,
d(U)>2^k \cdot \sum_{W\in\cB_0}d(W)\,\diam(W)^s},\] let
$\mathcal{D}_k=\myset{U}{U\in\mathcal{C}_k\mbox{ and no proper
superset of }U\mbox{ is in }\mathcal{C}_k}$. Then, using Property
\ref{cl4}, $\sum_{U\in\mathcal{D}_k}\diam(U)^s\le 2^{-k}$.

Notice that for every $k$, $\mathcal{D}_k$ is a $2^{-k/s}$-cover of
$\regSS[d]$, so $\dimh(A)\le s$. This completes our proof.

\end{proof}

Notice that a characterization of Hausdorff dimension in terms of
martingales (that is, 1-gales) that holds for Cantor space is not
clear for the general case. $d(U)$ is an $s$-gale if and only if
$d'(U)=\diam(U)^{s-1}d(U)$ is a 1-gale, but it is not clear how to
express $\regSS[d]$ in terms of $d'$.

For  the examples in subsection \ref{secnc}, we obtain
generalizations of the concept of gale used in previous
effectivizations.

\begin{example}

For the set $\Sigma^{\infty}$ of all infinite sequences over a
finite alphabet $\Sigma$ with the metric based on a positive and
nonatomic Borel probability measure $\nu: \Sigma^*\to [0,1]$.
(Generalization of \cite{DPSSF}).

For $s\in [0, \infty)$, an { $s$-supergale $d$} is a function $d:
\Sigma^*\to [0, \infty)$ such that for every $n\in\N$, for every
$w\in\Sigma^*$ with $|w|=n$ the following inequality holds
\begin{equation*}
d(w)\,\nu(w)^s\ge \sum_{a\in\Sigma} d(wa)\,\nu(wa)^s
\end{equation*}

\end{example}

\begin{example}

For $m\in\N$, consider the set $\R^m $ with the metric based on a strongly positive Borel  measure $\nu$. (Generalization
of \cite{DPSSF}).

For $s\in [0, \infty)$, an { $s$-supergale $d$} is a function $d: \Z^m\times \N^{m}\times \N\to [0, \infty)$
 such that for every $\mathbf{z}\in \Z^m, \mathbf{a}\in \N^{m},
n\in\N$ (with $a_i< 2^n$),  the following inequality holds
\begin{align}
d(\mathbf{z}, \mathbf{a}, n)\,\nu(I(\mathbf{z}, \mathbf{a}, n))^s\ge \label{eucG}\\
\sum_{I(\mathbf{t}, \mathbf{b}, n+1)\subseteq I(\mathbf{z}, \mathbf{a}, n)} d(\mathbf{t}, \mathbf{b},
n+1)\,\nu(I(\mathbf{t}, \mathbf{b}, n+1))^s.\nonumber
\end{align}

\end{example}

\begin{example}

Let $h:\N\to \N-\{0,1\}$,  $h^{\omega}=\Pi_{n\in\N}\{0, 1, \ldots,
h(n)-1\}$  with the metric based on a positive and nonatomic Borel
probability measure $\nu$. (Generalization of \cite{GreMil11}).

For $s\in [0, \infty)$, an { $s$-supergale $d$} is a function $d:
h^*\to [0, \infty)$ such that for every $n\in\N$, for every $w\in
h^n$ the following inequality holds
\begin{equation*}
d(w)\,\nu(w)^s\ge \sum_{a\in\Sigma} d(wa)\,\nu(wa)^s.\end{equation*}

\end{example}


\begin{example}
 $P_n(0,1)$ is the set of polynomials  with real coefficients and degree less than or equal to $n$, together with the
metric $d(f, g) = \|f-g\|_{\infty}$.

For $s\in [0, \infty)$, an { $s$-supergale $d$} is a function $d: \Z^{n+1}\times \N^{n+1}\times \N\to [0, \infty)$
 such that for every $\mathbf{z}\in \Z^{n+1}, \mathbf{a}\in \N^{n+1},
k\in\N$ (with $a_i< 2^k$),  the following inequality holds
\begin{equation}\label{eucP}
d(\mathbf{z}, \mathbf{a}, k)\ge  2^{-s} \sum_{DP(\mathbf{t}, \mathbf{b}, k+1)\subseteq DP(\mathbf{z}, \mathbf{a}, k)}
d(\mathbf{t}, \mathbf{b}, k+1).
\end{equation}

\end{example}

\section{Constructive dimension}

In this section we
effectivize Hausdorff dimension by considering constructive
dimension. We consider spaces that have computable nice covers
(defined below) and define constructive dimension in them.

Then we characterize constructive dimension in terms of Kolmogorov
complexity using the concept of Kolmogorov complexity of $x\in X$ at
precision $r\in \N$ inspired by \cite{DPSSF}. This characterization,
together with the absolute stability proven below allows a full
Theory of Information view of Hausdorff dimension in some general
metric spaces.

\begin{definition}
Let $X$ be a metric space with a nice cover
$(\mathcal{B}_n)_{n\in\N}$. Let  $\Sigma$ be finite and $\delta:
\Sigma^* \to \cB$ be surjective. We say that {\sl $(X, (\cB_n),
\delta)$ has a computable nice cover\/} if the following hold,
\begin{description}\setcounter{enumi}{3}
\item[(B5)] (Computable diameter) $\diam\circ \delta$ is a computable function.
\item[(B6)] (Computable cover) The function
\[P: \Sigma^* \times \N\to \Sigma^*\] defined  by $P(w, n)=<w_1,
\ldots, w_k>$ for $\delta(w)\in\cB_n$, such that \newline $\myset{V}{V\in\cB_{n+1}, V\subseteq \delta(w)}=\{\delta(w_1),
\ldots, \delta(w_k)\}$ is a computable function.
\end{description}
\end{definition}

For the rest of this section we fix a space $(X, (\mathcal{B}_n),
\delta)$ with a computable nice cover and omit $(\mathcal{B}_n),
\delta$ when referring to $X$.

\begin{definition} Let $d$ be a supergale. Then $d$ is contructive
if $d\circ\delta$ is lower semicomputable and $\sum_{U\in
\mathcal{B}_{0}} d(U)\,\diam(U)^s$ is a computable
number.\end{definition}

\begin{definition} Let $A\subseteq X$,
\[\hat{\mathcal{G}}_{\constr}(A)=\myset{s}{\mbox{there is a
constructive }s\mbox{-supergale }d \mbox{ with
}A\subseteq\regSS[d]}.\]
\end{definition}

\begin{definition} Let $A\subseteq X$. We define the constructive
dimension of $A$ as
$\cdim(A)= 
\inf\hat{\mathcal{G}}_{\constr}(A)$.\end{definition}

For each $x\in X$ we denote with $\cdim(x)$ the constructive
dimension of the singleton set, that is, $\cdim(x)=\cdim(\{x\})$.

Let us briefly comment on the importance of the choice of nice cover
when effectivizing a dimension. The classical  definition of
Hausdorff dimension is invariant under the choice of nice cover as
our characterization in Theorem \ref{thsg}\ proves. In the case of
effective dimension, however, this invariance does not generally
hold (the choice of $(\mathcal{B}_n), \delta$ can be relevant) and
this fact is very meaningful. For instance in the case of
Finite-State dimension in $\R^n $, that is, restriction to gales
that can be computed by Finite State Automata (done for Cantor space
in \cite{FSD}), it matters whether we use dyadic intervals, triadic
intervals, etc and this is related to the existence of normal
sequences that are not absolutely normal \cite{FSD}. On the other
hand certain invariance properties are known for constructive and
polynomial-time dimension \cite{HitMay13}.

We next look at the compression characterization of constructive
dimension. Notice that the definition of Kolmogorov complexity at a
certain precision is more involved than in the Cantor case (or even
slightly more than in the Euclidean case) when there is no Borel
measure.

Constructive dimension can be characterized in terms of Kolmogorov
complexity as follows. Let $\K(w)$ denote the usual self-delimiting
Kolmogorov complexity of $w\in \Sigma^*$.

\begin{definition} Let $x\in X$, let $r\in\N$. The {\sl Kolmogorov
complexity of $x$ at precision $r$\/} is
\[\K_r(x)=\inf\myset{\K(w)}{x\in\delta(w), \, \diam(\delta(w))\le
2^{-r}},\] with $\K_r(x)=\infty$ if not such $w$ exists.
\end{definition}

\begin{theorem} Let $X$ be a metric space with a computable nice cover. Let
$Z\subseteq X$,\[\cdim(Z)=\sup_{x\in Z}\liminf_r\frac{\K_r(x)}{r}.\]

\end{theorem}

\begin{proof}

The first direction uses semicomputability and  additivity
properties of prefix Kolmogorov complexity.

Let $s, s', s''$ be rational numbers such that $s>s'>s''>\sup_{x\in
Z}\liminf_r\frac{\K_r(x)}{r}$.
Let
\[A=\myset{w}{\K(w)\le -s'\log(\diam(\delta(w)))}.\] Then $A$ is computably enumerable.

We define $d$ as follows, let $U\in\cBn$ with $\diam(U)>0$,



\[d(U)= \sum_{V\subseteq U, V\in
\delta(A)\cap \cB_{\ge n}}\frac{\diam(V)^{s'}}{\diam(U)^{s}}.\]

$d$ is well defined since $\sum_{V\in\delta(A)}\diam(V)^{s'}\le
\sum_w 2^{-\K(w)}<\infty$. $d$ is constructible by property (B6).

$d$ is an $s$-supergale since for $W\in\cB_{n-1}$,
\[\sum_{\scriptstyle{U\subseteq W,
U\in\cB_{n}}} d(U)\, \diam(U)^s = \sum_{\scriptstyle{U\subseteq W,
U\in\cB_{n}}} \sum_{V\subseteq U, V\in \delta(A)\cap \cB_{\ge
n}}\diam(V)^{s'}\le \] \[\le \sum_{V\subseteq W, V\in \delta(A)\cap
\cB_{\ge n-1}}\diam(V)^{s'}= d(W) \diam(W)^s,\] where the last
property follows from property (A2).

If $U\in\delta(A)$ then $d(U)\ge \diam(U)^{s'-s}$. Let $x\in Z$.
Since $\K_r(x)< r s''$ for infinitely many $r$, for those $r$ there
is a $w_r$ with $\K(w_r)\le r s''$, $x\in\delta(w_r)$ and
$\diam(\delta(w_r))\le 2^{-r}$. Therefore $w_r\in A$ and
$d(\delta(w_r))\ge \diam(\delta(w_r))^{s'-s}\ge 2^{r(s-s')}$.

Therefore for each $x\in Z$ there is $(\delta(w_r))\subseteq\cB$
with $\limsup_r d(\delta(w_r))=\infty$ and by Property
\ref{players}\ $x\in\regSS[d]$. Therefore $Z\subseteq \regSS[d]$.

For the other direction, let $s>\cdim(Z)$. Let $d$ be a constructive
$s$-supergale such that $Z\subseteq \regSS[d]$. For each $k\in\N$,
let
\[A_k=\myset{w}{d(\delta(w))\ge 2^k
(\sum_{W\in\cB_0}d(W)\,\diam(W)^s)}.\] Then by Property \ref{cl4},
for each $r\in \N$ and for every $\mathcal{E}\subseteq\cB$ such that
all sets in $\mathcal{E}$ are incomparable, the number of $w\in
A_k\cap \delta^{-1}(\mathcal{E})$ such that $\diam(\delta(w))>
2^{-r-1}$ is at most $2^{-k+rs+s}$. Also by Property \ref{cl4}, if
$w\in A_k$ then $\diam(\delta(w))\le 2^{-k/s}$.

Fix $k\in\N$. We enumerate all strings $u$ in $A_k$ and include
$\delta(u)$ in $\mathcal{E}^k$  for the strings such that
$\delta(u)$ is incomparable with all sets previously included in
$\mathcal{E}^k$.

For $w\in A_k\cap \delta^{-1}(\mathcal{E}^k)$ there is an $r\ge
\lceil k/s-1\rceil$ with $2^{-r-1}<\diam(\delta(w))\le 2^{-r}$, and
$\K(w)\le rs+s-k+O(\log k)+O(\log r)$.

Since $x\in\regSS[d]$ for every $k$ there is $w\in A_k\cap
\delta^{-1}(\mathcal{E}^k)$ such that $x\in\delta(w)$ and therefore
\[\liminf_r\frac{\K_r(x)}{r}\le \frac{rs+s-k+O(\log k)+O(\log
r)}{r}\le s.\] Taking the supremum over all $x\in Z$ we have proven
our Theorem.
\end{proof}

\begin{corollary} Let $X$ be a metric space with a computable nice cover. Let $x\in
X$,\[\cdim(x)=\liminf_r\frac{\K_r(x)}{r}.\]

\end{corollary}

As a corollary we have the property of total stability of
constructive dimension (see \cite{DISS}\ for the corresponding
version in Cantor space).

\begin{corollary}  Let $X$ be a metric space with a computable nice
cover. Let $A\subseteq X$. Then \[\cdim(A)=\sup_{x\in A}\cdim(x).\]

\end{corollary}


We finish by presenting the constructive dimension of the spaces in
all previous examples. Notice that all the examples in subsection
\ref{secnc}\ have computable nice covers.

\begin{example}

For the set $\Sigma^{\infty}$ of all infinite sequences over a
finite alphabet $\Sigma$ with the metric based on a positive and
nonatomic Borel probability measure $\nu: \Sigma^*\to [0,1]$.
(Generalization of \cite{DPSSF}).

Constructive supergales are constructive functions $d: \Sigma^*\to
[0, \infty)$ that fulfill the supergale inequality.

For $x\in \Sigma^{\infty}$, the { Kolmogorov complexity of $x$ at
precision $r$} is
\[\K_r(x)=\inf\myset{\K(w)}{w\sqsubseteq x, \, \nu(w)\le
2^{-r}}.\]

\end{example}

\begin{example}

For  the set $\R^m $ ($m\in\N$) with the metric based on a strongly positive Borel  measure $\nu$. (Generalization of
\cite{DPSSF}).

Constructive supergales are constructive functions
 $d: \Z^m\times \N^{m}\times \N\to [0, \infty)$ that fulfill inequality (\ref{eucG}).

For $x\in\R^m$, the { Kolmogorov complexity of $x$ at precision $r$}
is \begin{eqnarray*}\K_r(x)=\inf\{\,\K(w)\,&|&\,w= <\mathbf{z}, \mathbf{a}, n>, \mathbf{z}\in\Z^m, \mathbf{a}\in\N^{m}, n\in\N, a_i< 2^n, x\in I(\mathbf{z}, \mathbf{a}, n),\\
&&\nu(I(\mathbf{z}, \mathbf{a}, n))\le 2^{-r}\}.\end{eqnarray*}

\end{example}

\begin{example}

Let $h:\N\to \N-\{0,1\}$,  $h^{\omega}=\Pi_{n\in\N}\{0, 1, \ldots,
h(n)-1\}$  with the metric based on a positive and nonatomic Borel
probability measure $\nu$. (Generalization of \cite{GreMil11}).

Constructive supergales are constructive functions  $d: h^*\to [0,
\infty)$ that fulfill the supergale inequality.

For $x\in h^{\omega}$, the { Kolmogorov complexity of $x$ at
precision $r$} is
\[\K_r(x)=\inf\myset{\K(w)}{w\mbox{ codifies }y\in h^*, y\sqsubseteq x, \, \nu(y)\le
2^{-r}}.\]

\end{example}

\begin{example}
$P_n(0,1)$ is the set of polynomials  with real coefficients and degree less than or equal to $n$, together with the metric
$d(f, g) = \|f-g\|_{\infty}$.

Constructive supergales are constructive functions
 $d: \Z^{n+1}\times \N^{n+1}\times \N\to [0, \infty)$ that fulfill inequality (\ref{eucP}).

For $f\in P_n(0,1)$, the { Kolmogorov complexity of $f$ at precision $r$} is
\begin{eqnarray*}\K_r(f)=\inf\{\,\K(w)\,&|&\,w= <\mathbf{z}, \mathbf{a}, k>, \mathbf{z}\in\Z^{n+1}, \mathbf{a}\in\N^{n+1},
k\in\N, a_i< 2^k, \\ &&f\in DP(\mathbf{z}, \mathbf{a}, k), k\ge r+\log(n+1)\}.\end{eqnarray*}
\end{example}

\section{Further directions}

This paper intends to give an initial view of effective dimension on
arbitrary metric spaces. A number of issues have not been addressed
here including the definition of resource-bounded dimension for
resource-bounds other than lower semicomputability and the role of
different (computable) nice covers in effectivization and conditions
for their equivalence within it.

\section*{Acknowledgment}

We thank Donald Stull for suggesting example \ref{polyn}.

\bibliographystyle{plain}
\bibliography{../biblios/todo}

\end{document}